\documentclass[conference,letterpaper]{IEEEtran}

\usepackage{geometry}
\usepackage[utf8]{inputenc}
\usepackage[T1]{fontenc}
\usepackage{url}
\usepackage{ifthen}
\usepackage{cite}
\usepackage[cmex10]{amsmath}
\usepackage{tikz}
\usepackage{amssymb}
\usepackage{amsthm,mathtools}

\usetikzlibrary{arrows.meta, positioning, calc, shapes.multipart}
\newtheorem{theorem}{Theorem}
\newtheorem{lemma}{Lemma}
\newtheorem{proposition}{Proposition}
\newtheorem{definition}{Definition}

\definecolor{grondblue}{RGB}{41, 128, 185}
\definecolor{osdorange}{RGB}{211, 84, 0}
\definecolor{nodebg}{RGB}{248, 249, 250}

\begin{document}
\title{Best-First Ordered Statistics Decoding \\ of Quantum LDPC Codes}

\author{
  \IEEEauthorblockN{Michele Banfi$^{*\dagger}$, Marco Ferrari$^{\dagger}$, Antonino Favano$^{*}$, Alberto Tarable$^{\dagger}$, and Luca Barletta$^{*}$}
  $^{\dagger}$Consiglio Nazionale delle
  Ricerche, Italy. E-mail: $\{$alberto.tarable, marcopietro.ferrari$\}$@cnr.it\\
  $^{*}$Politecnico di Milano, Milan, 20133, Italy. E-mail:  $\{$michele.banfi, antonino.favano, luca.barletta$\}$@polimi.it
}

\maketitle

\begin{abstract}
  Belief Propagation (BP) followed by Ordered Statistics Decoding (OSD) has emerged as the gold standard for decoding quantum low-density parity-check (QLDPC) codes. Recent advancements in this field have proposed new methods and algorithms to lower the complexity of this standard pipeline. Because of code degeneracy, and more in general because multiple distinct error patterns can produce the same syndrome, OSD is inherently a list-decoding technique; that is, it enumerates a set of syndrome-consistent candidates and returns the most probable one. In this work, we propose a variant of OSD, which we call Best-First OSD (BF-OSD), that explores the error-candidate space more efficiently by traversing it in order of decreasing likelihood, rather than by brute-force enumeration of a pre-selected subset. In addition, we depart from the conventional BP\,+\,OSD cascade: instead of conditioning the OSD invocation on BP convergence, we invoke OSD after a fixed, small number of BP iterations. This design choice is motivated by the full circuit-level noise regime, in which BP is particularly unreliable. Monte Carlo simulations of a family of Bivariate Bicycle (BB) codes under full circuit-level noise show that BF-OSD matches the performance of the BP\,+\,OSD baseline while exploring the solution space with 1/100th of the query budget.
\end{abstract}

\section{Introduction}

Recent work on Quantum Error Correcting Codes (QECC) \cite{vasic2025quantum} has highlighted the promising family of Quantum Low-Density Parity-Check Codes (QLDPC). QLDPC codes typically exploit the Calderbank--Shor--Steane (CSS) code construction \cite{calderbank1996good, steane1996multiple}. Because of the sparsity of their parity-check matrices, QLDPCs are appealing from a hardware perspective, as sparsity implies fewer and lower-weight connections between physical qubits. Particular attention has been given to one such construction, the Bivariate Bicycle (BB) codes \cite{bravyi2024high}, which allow for a practical physical layout on superconducting hardware.

The gold standard for decoding QLDPC codes combines two algorithms in cascade: Belief Propagation (BP) and Ordered Statistics Decoding (OSD) \cite{panteleev2021degenerate}. BP is run first, and OSD is invoked only if BP fails to converge to a valid syndrome-satisfying solution~\cite{roffe2020decoding}. In this setup, BP acts as a fast first \emph{guess} of the error pattern, while OSD serves as a fallback that is guaranteed to return a syndrome-valid solution. Two practical observations motivate our work. First, well-known variants of OSD used in this cascade, OSD-$w$ and OSD-CS \cite{panteleev2021degenerate, roffe2020decoding}, explore the solution space inefficiently, enumerating a pre-selected subset of low-reliability bit-flip combinations rather than traversing candidates in order of likelihood. Second, conditional invocation of OSD makes the overall decoder inefficient in regimes where BP often fails to converge due to a numerous presence of length-4 cycles, such as under the full circuit-level noise model considered in this work.

We therefore propose two intertwined changes to the standard decoding flow: (i) a new variant of OSD, Best-First OSD (BF-OSD), which traverses the affine coset of syndrome-consistent solutions in order of increasing cost via a best-first search over XOR combinations of null-space generators. Under a fixed query budget $Q$, BF-OSD achieves a tighter approximation to Maximum-Likelihood (ML) decoding than OSD-$w$ or OSD-CS at comparable cost, and is naturally anytime; that is, its solution improves as more queries are performed; and (ii) rather than gating OSD on BP outcome, we run BP for only a small number of iterations, enough to propagate local information across the Tanner graph and produce non-uniform soft outputs, and then always invoke BF-OSD. In support of this approach, we also replace the \emph{flooding} BP schedule with a \emph{serial} update scheme, which propagates information faster without increasing decoder complexity  \cite{ShaLit}.

These two design choices are coupled, reducing the number of BP iterations risks yielding posterior estimates that are far from convergence. To ensure we exploit all available soft information in this potentially degraded regime, we were motivated to investigate alternative OSD column-ordering conventions (Section~\ref{sec:ordering}) alongside the coset search carried out by BF-OSD (Section~\ref{sec:grond}). We evaluate the resulting decoding pipeline on a family of BB codes under the full circuit-level noise model.


\section{Noise Model}
\label{sec:noise}

As in \cite{bravyi2024high}, we adopt the \emph{full circuit-level} noise model, which captures realistic failure modes introduced by the quantum operations themselves (CNOTs and single-qubit gates, idling, state preparation and measurement), rather than assuming an ideal data-qubit depolarising channel. Crucially, this model also accounts for error propagation. For example, an $X$ error occurring on the control qubit before a CNOT gate propagates to the target. Finally, because measurements are themselves noisy, syndrome extraction must be repeated over multiple rounds.

As a consequence, the parity-check matrices $H_X$ and $H_Z$ from the CSS construction
\begin{equation}
  H =
  \begin{pmatrix}
    H_X & 0 \\
    0 &  H_Z
  \end{pmatrix},
  \qquad H_X H_Z^T = 0,
\end{equation}
no longer define the graph used for BP decoding. Instead, decoding is performed on spatio-temporal \emph{decoding matrices} $H_{X,\text{dec}}$ and $H_{Z,\text{dec}}$, whose columns index elementary circuit faults and whose rows index detect (syndrome-difference) events. The decoder task is to identify which fault, or which linear combination of faults, occurred during computation, given the observed detector pattern.

The construction of these matrices, as described in \cite{bravyi2024high, dennis2002topological}, starts by building the stabilizer circuit that implements one round of syndrome extraction (with qubit connectivity defined by $H_X$ and $H_Z$), and then repeating it for $m$ rounds. We follow the standard convention of setting the number of syndrome‑extraction rounds equal to the code distance, $m = d$, thereby balancing the time‑like and space‑like components of the circuit‑level fault distance \cite{dennis2002topological,fowler2012surface, bravyi2024high}.

On top of this space-time circuit, we enumerate every single-fault location. For each fault $f$, we simulate error propagation through the remainder of the circuit, recording the tuple formed by the \emph{detector pattern} $d_f$ and the \emph{logical effect} $\ell_f$. Entries with identical tuples are merged and their fault probabilities summed, since two faults that share both syndrome and logical effect are, by construction, indistinguishable to any decoder. In the end,
\begin{equation}
  p(d, \ell) \;=\; \sum_{f:\, (d_f,\ell_f)=(d,\ell)} p_f,
\end{equation}
where $p(d, \ell)$ is the probability of the tuple $(d, \ell)$ and $p_f$ is the probability of fault $f$.

Tracking the logical effect, not just the syndrome, is essential in QEC, since any correction that commutes with the logical operators yields a valid recovery. This is the familiar notion of \emph{degeneracy}, which is typical of quantum codes and has no classical counterpart. Even when the decoder outputs a physical error $e_2$ that differs from the true error $e_1$, as long as $e_1 e_2$ belongs to the stabilizer group $ \mathcal{S}$, the logical state is preserved, $e_1 e_2 |\psi\rangle \;=\; S|\psi\rangle \;=\; |\psi\rangle,\quad S\in\mathcal{S}.$ In practice, this means that the list of candidate errors explored by OSD can safely include many \emph{stabilizer-equivalent} alternatives, and the true quantity of interest is the logical failure probability rather than the disagreement with the true physical error.

\subsection*{Performance Metric}
Following \cite{bravyi2024high}, we report the per-round logical error rate $p_L$, obtained from the per-shot logical failure probability $P_L(N_c)$ measured over $N_c$ syndrome-extraction rounds:
\begin{equation}
  p_L \;=\; 1 - \bigl(1 - P_L(N_c)\bigr)^{1/N_c} \;\approx\; P_L(N_c)/N_c,
\end{equation}
where the approximation holds in the regime $P_L(N_c)\ll 1$. Reporting $p_L$ rather than $P_L(N_c)$ enables fair comparisons across codes simulated for different numbers of rounds and makes the threshold and sub-threshold scaling directly visible.

\subsection*{Implication for the decoder}
The full circuit-level noise model has two key consequences that shape our decoder design. First, the decoding matrices $H_{X,\text{dec}}$ and $H_{Z,\text{dec}}$ are substantially denser and larger than the original CSS check matrices; consequently, BP operates on a Tanner graph with many more short cycles. As a result, BP convergence is slower and less reliable than in a code-capacity setting. 
Second, the merged fault probabilities tend to be relatively uniform across the columns of $H_{X,\text{dec}}$ and $H_{Z,\text{dec}}$. However, even a small number of BP iterations can produce informative, non-uniform soft outputs that OSD can exploit. Together, these observations justify our design choices: we use a short BP \emph{warm-up} phase sufficient to inform the subsequent OSD stage, without attempting to drive BP to full convergence.

\section{Standard OSD Implementations}
\label{sec:standard-osd}

Two variants of OSD are commonly used in the literature: OSD-$w$ (also called OSD-E, where E stands for \emph{Exhaustive}) and OSD-CS (where CS stands for \emph{Combination Sweep})~\cite{roffe2020decoding}. Both build on a common foundation, OSD-0, and then search for a more likely candidate by enumerating bit-flip combinations over the OSD-0 free set. They differ only in which combinations are explored.

\subsection{OSD-0}
OSD-0 is the core of standard OSD approaches and also the starting point of our algorithm. It is guaranteed to return a syndrome-consistent error vector $e$. Given the BP soft outputs (LLRs), OSD-0 reorders the columns of the decoding matrix $H_{\text{dec}}$ and applies Gaussian Elimination (GE) to select $r = \text{rank}(H_{\text{dec}})$ linearly independent pivot columns $\mathcal{P}$. The remaining $k = N - r$ columns form the free set $\mathcal{F}$.

Gaussian elimination is guaranteed to select an independent pivot set of minimum aggregate LLR weight, i.e., an optimal solution obtained via a greedy algorithm (see the Appendix for a formal proof). We note that this property is often overlooked in the literature and state it here explicitly.
\

The free bits are then assigned their hard decisions, and the pivot bits are solved from the syndrome via back-substitution, as discussed in Section~\ref{sec:ordering}. The output is a candidate error $E_{\text{base}}$ that satisfies $H \cdot E_{\text{base}} = s.$

Two distinct column-ordering conventions appear in the literature and differ significantly in determining which bits become pivots. We discuss this trade-off and our choice in Section~\ref{sec:ordering}. On its own, OSD-0 does not approach ML performance; to do so, one must explore the solution space around $E_{\text{base}}$, which nevertheless provides a good starting point.

\subsection{OSD-$w$}
OSD-$w$~\cite{roffe2020decoding} selects the $w$ least reliable bits in $\mathcal{F}$ and enumerates all possible bit-flip patterns on them, producing a search space $\mathcal{D}_w$ of size
\begin{equation}
  |\mathcal{D}_w| \;=\; 2^w.
  \label{eq:osdw-size}
\end{equation}
Among these candidates, including $E_{\text{base}}$, the one with the lowest cost~\eqref{eq:cost} is returned. The original classical formulation of Fossorier and Lin~\cite{fossorier1995soft} instead enumerates all error patterns of Hamming weight up to $w$ over $\mathcal{F}$, yielding a search space of size $\sum_{i=0}^{w}\binom{k}{i}$. The quantum adaptation restricts the search to the least reliable subset of $\mathcal{F}$ to maintain tractability for the large information sets typical of QLDPC codes.

\subsection{OSD-CS}
OSD-CS is a fast heuristic introduced by Roffe \emph{et al.}~\cite{roffe2020decoding}. Rather than enumerating subsets of $\mathcal{F}$ up to weight $w$, it splits the search into two phases:
\begin{itemize}
  \item {Weight-1 phase:} test each of the $k$ free columns individually ($k$ candidates).
  \item {Weight-2 phase:} test all pairs drawn from the $\lambda$ \emph{most unreliable} free columns (i.e., those with the smallest LLR magnitude), yielding $\binom{\lambda}{2}$ candidates.
\end{itemize}
The resulting search space has size
\begin{equation}
  |\mathcal{D}_{\text{CS}}| \;=\; k + \binom{\lambda}{2}.
\end{equation}
Bravyi \emph{et al.}~\cite{bravyi2024high} use $\lambda=7$ in their simulations of BB codes under full circuit-level noise. While OSD-CS is computationally efficient and performs well in practice, it is inherently \emph{a priori}, i.e., it selects a fixed subregion of the search space regardless of how  candidates actually score. BF-OSD, introduced next, is designed to address this limitation.

\section{On OSD-0 Column Ordering}
\label{sec:ordering}

A design choice in OSD-0 that is rarely discussed explicitly is how columns are sorted before GE. This choice determines which bits become pivots and which remain free. Two conventions appear in the literature; we review both and argue that, in our setting where BP runs for only a few iterations, the classical formulation of Fossorier, properly adapted to syndrome decoding, provides a better baseline for BF-OSD solution-space exploration.

Throughout this section, we define the LLR as
\begin{equation} \label{eq:LLRs}
  \Lambda_i = \log\frac{P(e_i=0)}{P(e_i=1)},
\end{equation}
so that $\Lambda_i > 0$ indicates that BP believes bit $i$ is likely correct, $\Lambda_i < 0$ indicates that BP believes it is likely an error, and $|\Lambda_i|$ measures confidence regardless of sign.

\subsection{LLR Convention}

Roffe et al.\ \cite{roffe2020decoding} sort columns by ascending LLR, equivalently by descending $P(e_i = 1)$. Thus, bits most likely to be in error are included first in the pivot set $\mathcal{P}$, while bits with the largest positive $\Lambda_i$ are placed in the free set $\mathcal{F}$ and set to zero.

This choice has a potential drawback. When BP is confident that $e_i=1$ ($\Lambda_i \ll 0$), that hard decision is already available, yet it occupies a pivot position whose value is recomputed via matrix inversion. Meanwhile, genuinely uncertain bits (small $|\Lambda_i|$) may end up in the free set and be forced to zero, even though these are precisely the bits whose values should be inferred from the syndrome.

\subsection{Confidence Convention}

Fossorier and Lin \cite{fossorier1995soft} sort columns by ascending $|\Lambda_i|$, so that the least reliable bits (regardless of sign) become pivots, while the most reliable become free. In the original classical (codeword) setting, a preprocessing step flips every bit whose hard decision is 1, updates the syndrome accordingly, and negates the corresponding LLR, ensuring that all LLRs are positive. Free bits can then safely be set to zero, and pivot bits are solved from the updated syndrome.

In the quantum setting, although there is no received codeword, this preprocessing step has a natural interpretation in syndrome decoding: for every bit $i$ with $\Lambda_i < 0$, we accept the BP hard decision ($e_i = 1$), absorb its effect into the syndrome ($s \leftarrow s \oplus H_{\text{col}_i}$), and replace $\Lambda_i \leftarrow |\Lambda_i|$. After this transformation, all LLRs are non-negative, confidence sorting reduces to a standard ascending order, and OSD-0 proceeds exactly as in the classical case. 

As a result, uncertain bits are solved from the syndrome, where they belong, while confident bits, whether correct or erroneous, retain their hard decisions. This avoids both failure modes: no confident-error bit is forced to zero, and no uncertain bit is locked into a default value.

To our knowledge, this distinction has not been made explicit in the QLDPC literature.

\subsection{Implications for BF-OSD}

In the standard BP\,+\,OSD pipeline, BP is typically run until convergence (or until a maximum number of iterations is reached), and OSD is invoked only upon failure. In that regime, most posterior LLRs are positive after convergence, and the difference between the two ordering conventions is minor. In our setting, however, BP is run for only a few iterations to propagate local information across the Tanner graph. As a result, the posterior LLRs are often far from convergence.

By adopting the confidence-based ordering with the pre-flip step, BF-OSD produces an OSD-0 baseline $E_{\text{base}}$ that respects all BP hard decisions  while concentrating the degrees of freedom on  genuinely uncertain bits. This provides a stronger starting point for the subsequent null-space coset search.

\section{BF-OSD}
\label{sec:grond}

BF-OSD is built on top of the OSD-0 baseline of Section~\ref{sec:ordering}: confidence sorting with the pre-flip step adapted to syndrome decoding. It takes $E_{\text{base}}$ as a starting point and explores the solution space of $H_{\text{dec}} \cdot e = s$ in order of increasing cost, or equivalently decreasing likelihood, until a predefined query budget $Q$ is exhausted.

\subsection{The Solution Space as an Affine Coset}

The set of bit-strings consistent with the observed syndrome is an affine subspace of $\mathbb{F}_2^{N}$, namely a coset of $\ker(H_{\text{dec}})$:
\begin{equation}
  \{e \in \mathbb{F}_2^{N} : H_{\text{dec}} \cdot e = s\} \;=\; E_{\text{base}} \,\oplus\, \ker(H_{\text{dec}}).
\end{equation}
Every candidate can therefore be written as
\begin{equation}
  e \;=\; E_{\text{base}} \,\oplus\, g, \qquad g \in \ker(H_{\text{dec}}),
\end{equation}
and the syndrome is satisfied by construction, since $H_{\text{dec}}\, g = 0$ implies $H_{\text{dec}}\, e = H_{\text{dec}}\, E_{\text{base}} = s$. The kernel has dimension $N-r$, where $r=\text{rank}(H_{\text{dec}})$ and $k=N-r$ is the size of the free set.

A natural basis of $\ker(H_{\text{dec}})$ can be read directly from the reduced row-echelon form $R$ produced by GE.
Let $\mathcal{P} = \{\pi_1, \dots, \pi_r\}$ and $\mathcal{F} = \{f_1, \dots, f_k\}$ denote the pivot and free column indices, respectively.
For each free column $j \in \mathcal{F}$, define the null-space generator $g_j \in \mathbb{F}_2^{N}$ component-wise as
\begin{equation}
  (g_j)_\ell \;=\;
  \begin{cases}
    1        & \text{if } \ell = j, \\
    R_{i,j}  & \text{if } \ell = \pi_i \text{ for some } i \in \{1,\dots,r\}, \\
    0        & \text{otherwise}.
  \end{cases}
  \label{eq:generator}
\end{equation}

\subsection{Cost Function and Generator Weights}

We rank candidates by their negative log-likelihood cost. Using the LLRs~\eqref{eq:LLRs}, we define
\begin{equation}
  \text{cost}(e) \;=\; \sum_{i:\, e_i = 1} |\Lambda_i|,
  \label{eq:cost}
\end{equation}
so that lower cost corresponds to a more probable error pattern. After the confidence pre-flip, all LLRs are non-negative, so $|\Lambda_i| = \Lambda_i$ and the cost reduces to a simple sum over the support of $e$. Since OSD-0 selects pivots optimally, $E_{\text{base}}$ already attains a good cost within the coset; BF-OSD then searches for a non-trivial $g \in \ker(H_{\text{dec}})$ such that $E_{\text{base}} \oplus g$ has strictly lower cost.

To order the search, we associate each generator $g_j$ with a soft weight
\begin{equation}
  w_j \;=\; \sum_{i:\, (g_j)_i = 1} \Lambda_i,
  \label{eq:gen-weight}
\end{equation}
i.e., the cost incurred by XOR-ing $g_j$ alone into $E_{\text{base}}$. Generators with low $w_j$ act on bits about which BP was uncertain, so XOR-ing them with $E_{\text{base}}$ is \emph{cheap}. Generators with high $w_j$ would overturn reliable hard decisions and are \emph{expensive}.

\subsection{Best-first Coset Search}

BF-OSD explores the coset via a best-first search tree (see the example in Figure~\ref{fig:grond_tree}), systematically evaluating combinations of generators in increasing order of cost. The state space of this search tree is $\ker(H_{\text{dec}})$. A node in the tree is defined by a subset of free column indices $T \subseteq \mathcal{F}$, which corresponds to the candidate error $e_T$ which costs is evaluated using~\eqref{eq:cost}

The search begins at the root node, represented by the empty set $T = \emptyset$, which corresponds to the baseline $e_\emptyset = E_{\text{base}}$.

To expand the tree without generating duplicate combinations, we enforce a strict ordering rule. The children of a node $T$ are generated by adjoining a single new generator $g_{j^\star}$, with the condition that $j^\star > \max T$. This guarantees that each combination is visited exactly once.

To ensure the algorithm evaluates the most probable candidates first, BF-OSD employs a priority queue keyed on the candidate's total cost defined in~\eqref{eq:cost}. The priority queue acts as a sorted waiting list, always keeping the lowest-cost candidate at the front.

Rather than computing this sum from scratch for every node, BF-OSD computes the cost of a child node $T' = T \cup \{j^\star\}$ incrementally. The new candidate $e_{T'} = e_T \oplus g_{j^\star}$ differs from its parent $e_T$ only at the bit indices where $(g_{j^\star})_i = 1$. Therefore, BF-OSD updates the parent cost by considering only these flipped bits: if bit $i$ flips in $e_T$, the cost increases or decreases by $|\Lambda_i|$, depending on whether the bit changes from $0$ to $1$ or from $1$ to $0$.
By pushing these incrementally evaluated children into the priority queue and always popping the top element, BF-OSD traverses the candidate space in non-decreasing order of cost.


\begin{figure}[!t]
  \centering
  \resizebox{\linewidth}{!}{%
    \begin{tikzpicture}[
        >=Latex,
        treenode/.style={
          rectangle, rounded corners=4pt, draw=black!40, thick,
          fill=nodebg, align=center, inner sep=8pt, minimum width=2.4cm,
          font=\small
        },
        grondbadge/.style={
          circle, fill=grondblue!10, draw=grondblue, thick,
          inner sep=1pt, minimum size=16pt, font=\sffamily\bfseries\footnotesize, text=grondblue
        },
        osdbadge/.style={
          circle, fill=osdorange!10, draw=osdorange, thick,
          inner sep=1pt, minimum size=16pt, font=\sffamily\bfseries\footnotesize, text=osdorange
        },
        skipbadge/.style={
          circle, fill=black!10, draw=black!50, thick,
          inner sep=1pt, minimum size=16pt, font=\sffamily\bfseries\footnotesize, text=black!70
        },
        edge/.style={
          draw=black!50, thick, ->
        }
      ]

      \draw[dashed, gray!40, thick] (-7.5, -2.5) -- (5.5, -2.5) node[right, font=\itshape\small, text=gray!70!black] {Weight 1};
      \draw[dashed, gray!40, thick] (-7.5, -5.0) -- (5.5, -5.0) node[right, font=\itshape\small, text=gray!70!black] {Weight 2};
      \draw[dashed, gray!40, thick] (-7.5, -7.5) -- (5.5, -7.5) node[right, font=\itshape\small, text=gray!70!black] {Weight 3};


      \node[treenode] (root) at (0, 0) {$T = \emptyset$\\$\text{cost} = 0$};

      \node[treenode] (n1) at (-4, -2.5) {$T = \{1\}$\\$\text{cost} = 3$};
      \node[treenode] (n2) at (0, -2.5) {$T = \{2\}$\\$\text{cost} = 6$};
      \node[treenode] (n3) at (4, -2.5) {$T = \{3\}$\\$\text{cost} = 9$};

      \node[treenode] (n12) at (-6, -5) {$T = \{1, 2\}$\\$\text{cost} = 4$};
      \node[treenode] (n13) at (-2, -5) {$T = \{1, 3\}$\\$\text{cost} = 7$};
      \node[treenode] (n23) at (2, -5) {$T = \{2, 3\}$\\$\text{cost} = 10$};

      \node[treenode] (n123) at (-6, -7.5) {$T = \{1, 2, 3\}$\\$\text{cost} = 8$};

      \draw[edge] (root) -- (n1) node[pos=0.2, left=10pt, text=black!60] {$\oplus g_1$};
      \draw[edge] (root) -- (n2) node[midway, left=2pt,  text=black!60] {$\oplus g_2$};
      \draw[edge] (root) -- (n3) node[pos=0.2, right=10pt,  text=black!60] {$\oplus g_3$};

      \draw[edge] (n1) -- (n12) node[midway, left=2pt,  text=black!60] {$\oplus g_2$};
      \draw[edge] (n1) -- (n13) node[midway, right=2pt,  text=black!60] {$\oplus g_3$};
      \draw[edge] (n2) -- (n23) node[midway, right=2pt,  text=black!60] {$\oplus g_3$};

      \draw[edge] (n12) -- (n123) node[midway, left=2pt,  text=black!60] {$\oplus g_3$};

      \node[grondbadge, xshift=-10pt, yshift=10pt] at (root.north west) {1};
      \node[grondbadge, xshift=-10pt, yshift=10pt] at (n1.north west) {2};
      \node[grondbadge, xshift=-10pt, yshift=10pt] at (n12.north west) {3}; 
      \node[grondbadge, xshift=-10pt, yshift=10pt] at (n2.north west) {4};
      \node[grondbadge, xshift=-10pt, yshift=10pt] at (n13.north west) {5};
      \node[grondbadge, xshift=-10pt, yshift=10pt] at (n123.north west) {6}; 
      \node[grondbadge, xshift=-10pt, yshift=10pt] at (n3.north west) {7};
      \node[skipbadge, xshift=-10pt, yshift=10pt] at (n23.north west) {X};

      \node[osdbadge, xshift=10pt, yshift=10pt] at (root.north east) {1};
      \node[osdbadge, xshift=10pt, yshift=10pt] at (n1.north east) {2};
      \node[osdbadge, xshift=10pt, yshift=10pt] at (n2.north east) {3};
      \node[osdbadge, xshift=10pt, yshift=10pt] at (n3.north east) {4}; 
      \node[osdbadge, xshift=10pt, yshift=10pt] at (n12.north east) {5};
      \node[osdbadge, xshift=10pt, yshift=10pt] at (n13.north east) {6};
      \node[osdbadge, xshift=10pt, yshift=10pt] at (n23.north east) {7};
      \node[skipbadge, xshift=10pt, yshift=10pt] at (n123.north east) {X}; 

      \begin{scope}[shift={(-8, -0.75)}]
        \draw[rounded corners=4pt, fill=white, draw=black!30, thick] (0,0) rectangle (5.6, 2.0);
        \node[font=\bfseries\small, anchor=north] at (2.8, 2.0) {Evaluation Order \& Setup};

        \node[grondbadge] at (0.6, 1.2) {\#};
        \node[anchor=west, font=\small] at (1.0, 1.2) {BF-OSD (Cost-driven)};

        \node[osdbadge] at (0.6, 0.5) {\#};
        \node[anchor=west, font=\small] at (1.0, 0.5) {OSD-CS (Breadth-first)};
      \end{scope}

    \end{tikzpicture}%
  }
  \caption{Comparison of solution-space traversal between BF-OSD and OSD-CS on a subset of three generators with initial base costs $w(g_1)=3, w(g_2)=6$, and $w(g_3)=9$. Because generators $\{1\}$ and $\{2\}$ destructively interfere to lower the total cost, BF-OSD's priority queue immediately traverses to depth 2 (evaluating $\{1,2\}$ as its $3^{\text{rd}}$ step) before evaluating higher-cost single-generator candidates. By constrast, OSD-CS statically sweeps all weight-1 combinations before moving to weight-2 combinations, wasting query budget on low-probability candidates and completely missing the weight-3 candidate. Moreover, with a budget $Q = 7$, the node with cost $= 10$ is never visited by BF-OSD.}
  \label{fig:grond_tree}
\end{figure}
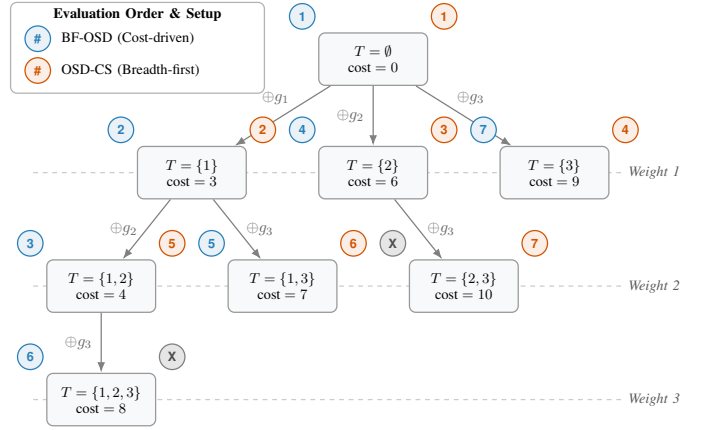

\subsection{Budget, Termination, and Anytime Behaviour}

Because every candidate $E_{\text{base}} \oplus g$ already satisfies the syndrome, there is no early stopping based on syndrome matching. Instead, BF-OSD evaluates candidates one at a time until its query budget $Q$ is exhausted and returns the lowest-cost candidate seen so far. Several properties follow directly:
\begin{itemize}
  \item Monotone improvement: the returned cost can only decrease, or remain constant, as $Q$ increases. The first pop candidate is $e_\emptyset = E_{\text{base}}$, and any subsequent candidate is retained only if it improves the running best.
  \item No early termination by weight: a weight-$2$ generator combination can have lower cost than every weight-$1$ combination (when two generators share many reliable pivots). Therefore, the heap must be drained up to $Q$.
  \item Anytime decoding: BF-OSD always holds a valid, syndrome-consistent candidate. It can be interrupted at any moment, returning the current best solution.
  \item Coverage: in the limit $Q \to 2^{k}$, BF-OSD enumerates the entire coset and coincides with coset ML decoding.
\end{itemize}

\section{Performance Analysis}
\label{sec:complexity}

Let $N$ be the number of columns of $H_{\text{dec}}$, $r=\text{rank}(H_{\text{dec}})$, $k=N-r$ the size of the free set, $I_{\text{BP}}$ the number of BP iterations, and $Q$ the BF-OSD query budget.

\paragraph*{Belief Propagation} BP performs message updates along each edge of the Tanner graph, so its complexity is $\mathcal{O}(|E|)$ per iteration, where $|E|$ is the edge count. Replacing the flooding schedule with a sequential one leaves this per-iteration cost unchanged: the same message updates are performed, only in a different order. However, the sequential schedule substantially reduces the number of iterations needed to propagate information over a given graph distance, allowing $I_{\text{BP}}$ to remain small. The total BP cost is therefore $\mathcal{O}(I_{\text{BP}} |E|)$.

\paragraph*{OSD-0} Column sorting costs $\mathcal{O}(N \log N)$, and GE over $\mathbb{F}_2$ on an $r \times N$ matrix is $\mathcal{O}(r^2 N)$.

\paragraph*{Solution-space exploration} The three OSD variants differ only in how they enumerate candidates on top of the same OSD-0 baseline. OSD-$w$ evaluates $|\mathcal{D}_w| = 2^w$ candidates over the $w$ least reliable free columns, each with cost $\mathcal{O}(N)$ for the XOR and cost evaluation, giving a total exploration cost of $\mathcal{O}(2^w N)$. OSD-CS evaluates $|\mathcal{D}_{\text{CS}}| = k + \binom{\lambda}{2}$ candidates at the same per-candidate cost, for a total of $\mathcal{O}((k + \lambda^2)\, N)$. BF-OSD evaluates $Q$ candidates, with an additional amortized $\mathcal{O}(\log Q)$ cost per pop for the priority queue, yielding a total exploration cost of $\mathcal{O}(Q(N + \log Q))$. Choosing $Q \lesssim |\mathcal{D}_{\text{CS}}|$ therefore gives BF-OSD a budget comparable to OSD-CS, while the best-first ordering ensures that the $Q$ evaluated candidates are the $Q$ lowest-cost ones in the coset.

\subsection{Simulation Results}
\label{sec:results}

We benchmark the full BP\,+\,BF-OSD pipeline on the same family of BB codes reported in~\cite{bravyi2024high}, namely those reported in Table~\ref{tab:results_ttable}. For each code, we simulate $N_c = d$ rounds of syndrome extraction under full circuit-level noise at a physical error rate of $p = 10^{-3}$, matching the setup of~\cite{bravyi2024high}. BP is run using the sum-product algorithm with a  sequential schedule for $I_{\text{BP}}$ iterations, where $I_{\text{BP}}$ is much smaller than the convergence horizon of $10000$ iterations used in~\cite{bravyi2024high}.  BF-OSD is then always invoked with query budget $Q$. For reference, we also run with the same simulator on identical shots using the baseline BP\,+\,OSD-CS decoder of~\cite{bravyi2024high} with $\lambda=7$. We report the per-round logical error rate $p_L$, obtained from $P_L(N_c)$ via the conversion introduced in Section~\ref{sec:noise}.

Figure~\ref{fig:results} compares the logical error rate of BF-OSD with that of the baseline across the five BB codes, while Table~\ref{tab:results_ttable} reports the corresponding values at $p = 10^{-3}$. The expected BB-code behavior is observed: $p_L$ decreases with increasing code distance.  BF-OSD tracks, and in some cases slightly improves upon the BP\,+\,OSD-CS baseline while using a comparable query budget.

We further tested the same codes with a query budget $Q$ of 1\% of the original budget, obtaining comparable results.

Finally, we tested the above experiments with both column-ordering conventions and did not observe statistically meaningful differences. This suggests that column ordering may have only a limited effect on BF-OSD in the present regime. In the discussion of future work, we formulate hypotheses for regimes in which this ordering may become more important.

\begin{figure}
  \centering
  \includegraphics[width=0.85\linewidth]{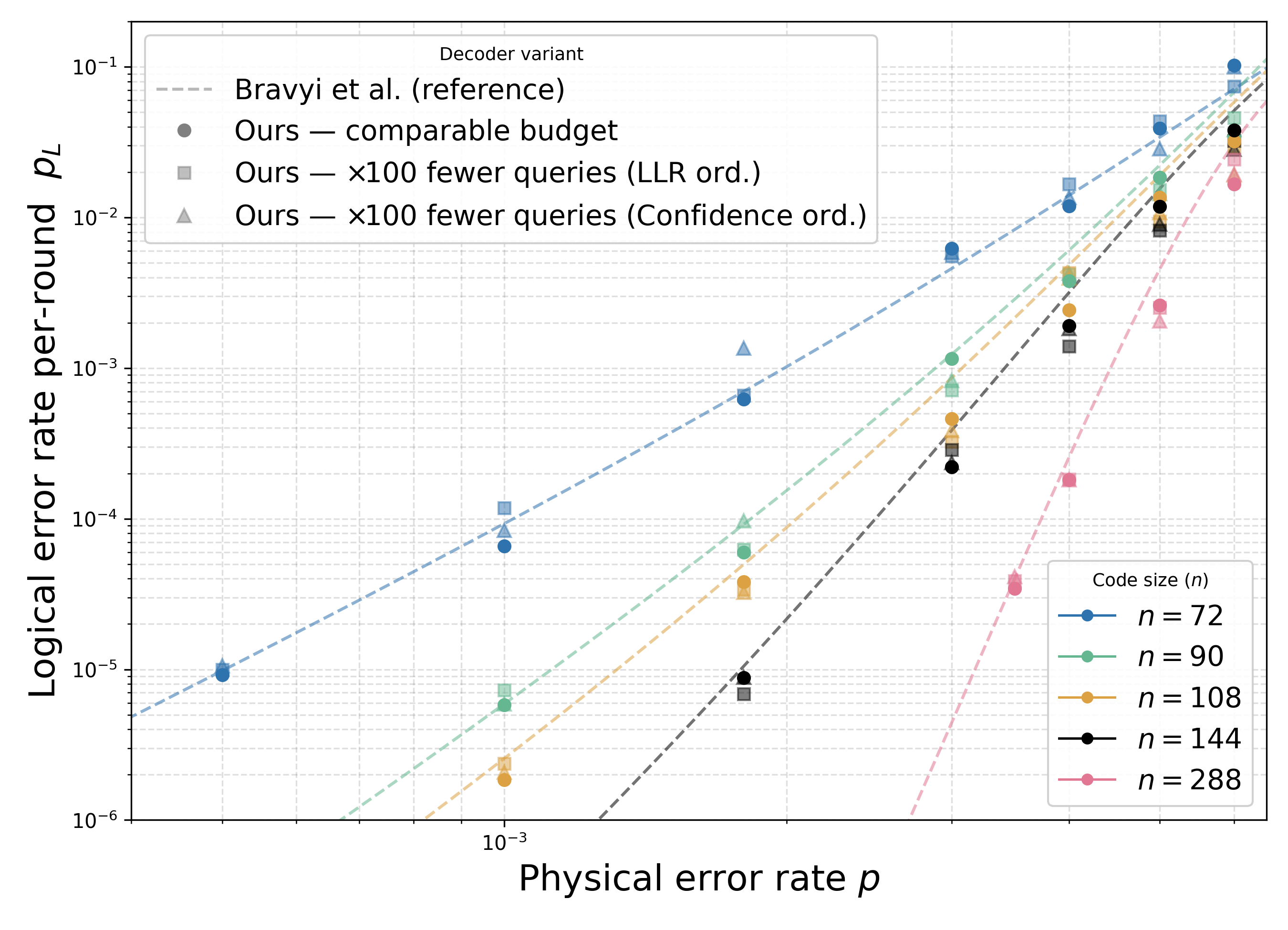}
  \caption{Per-round logical error rate $p_L$ of the BP\,+\,BF-OSD decoder under full circuit-level noise, for the five BB codes of~\cite{bravyi2024high}, as a function of physical error rate. Dashed curves report the BP\,+\,OSD-CS baseline of~\cite{bravyi2024high} for reference. }
  \label{fig:results}
\end{figure}

\begin{table}
  \renewcommand{\arraystretch}{1.2}
  \caption{Per-round logical error rate $p_L$ at physical error rate $p=10^{-3}$ under full circuit-level noise. Baseline values are from~\cite{bravyi2024high}.}
  \label{tab:results_ttable}
  \centering
  \begin{tabular}{lcc}
    \hline \hline
    \textbf{Code $[[N, k, d]]$} & \textbf{BP\,+\,OSD-CS~\cite{bravyi2024high}} & \textbf{BP\,+\,BF-OSD (this work)} \\
    \hline
    $[[72, 12, 6]]$   & $7 \times 10^{-5}$  & $6.58 \times 10^{-5}$ \\
    $[[90, 8, 10]]$   & $5 \times 10^{-6}$  & $5.81 \times 10^{-6}$ \\
    $[[108, 8, 10]]$  & $2 \times 10^{-6}$  & $1.85 \times 10^{-6}$ \\
    $[[144, 12, 12]]$ & $2 \times 10^{-7}$  & $2.38 \times 10^{-7}$ \\
    $[[288, 12, 18]]$ & $2 \times 10^{-12}$ & $8.17 \times 10^{-12}$ \\
    \hline \hline
  \end{tabular}
\end{table}

\section{Conclusion}

We have presented Best-First Ordered Statistic Decoding (BF-OSD), a variant of OSD that explores the coset of syndrome-consistent solutions in order of increasing cost via a best-first traversal of XOR combinations of null-space generators. Under a fixed query budget $Q$, BF-OSD uses that budget more efficiently than OSD-$w$ or OSD-CS, as it evaluates the $Q$ lowest-cost candidates according to an exact, incrementally maintained cost metric. We have further argued that, in the full circuit-level noise setting, conditioning the invocation of OSD on BP convergence is counterproductive. BP converges slowly on the denser space-time Tanner graph, yet even a small number of sequential iterations already produces informative soft outputs that BF-OSD can exploit effectively. Our simulations on the Bivariate Bicycle codes of~\cite{bravyi2024high} support this interpretation. Natural directions for future work include designing an adaptive query budget $Q(s, \Lambda)$ conditioned on the observed syndrome and BP reliability, as well as studying BF-OSD under min-sum approximations. The latter would help clarify whether, in hardware friendly implementations, the choice of OSD column ordering has a measurable impact on the logical error rate.

\appendix
\label{sec:appendix}
Applying Gaussian elimination to a matrix whose columns are sorted by weight, effectively acts as a greedy algorithm: a column is selected as a pivot if and only if it is linearly independent of the previously selected pivot columns.

\begin{definition}
  Let $E$ be a finite set, and let $\mathcal{I}\subseteq 2^E$. The pair $(E,\mathcal{I})$ is a \emph{matroid} if the following three properties hold:
  \begin{enumerate}
    \item $\varnothing\in \mathcal{I}$;
    \item whenever $I\in \mathcal{I}$ and $J\subseteq I$, one has $J\in \mathcal{I}$;
    \item whenever $I,J\in \mathcal{I}$ and $|I|<|J|$, there exists $e\in J\setminus I$ such that $I\cup\{e\}\in \mathcal{I}$.
  \end{enumerate}
  The sets in $\mathcal{I}$ are called \emph{independent}.
\end{definition}

For a matroid $(E,\mathcal{I})$ and a subset $X\subseteq E$, define the \emph{rank} of $X$ by
\begin{equation}
  \rho(X)\triangleq \max\{ |I| : I\subseteq X,\ I\in \mathcal{I} \}.
\end{equation}
A subset $B\subseteq E$ is called a \emph{basis} if $B\in \mathcal{I}$ and $|B|=\rho(E)$.

\begin{lemma}
  \label{lem:maximal-equals-rank}
  Let $(E,\mathcal{I})$ be a matroid, and let $X\subseteq E$. If $I\subseteq X$ is independent and maximal with respect to inclusion among independent subsets of $X$, then $|I|=\rho(X)$.
\end{lemma}

\begin{proof}
  Since $I\subseteq X$ is independent, by definition of rank one has $|I|\le \rho(X)$. Suppose, for contradiction, that $|I|<\rho(X)$. Then there exists an independent set $J\subseteq X$ such that $|J|=\rho(X)> |I|$. By the augmentation axiom of the matroid, there exists $e\in J\setminus I$ such that $I\cup\{e\}$ is independent. Because $J\subseteq X$, the element $e$ belongs to $X$, and hence $I\cup\{e\}\subseteq X$. This contradicts the maximality of $I$ among independent subsets of $X$. Therefore $|I|=\rho(X)$.
\end{proof}

\begin{proposition}
  \label{prop:vector-matroid}
  Let $A=[\mathbf{a}_1\,\mathbf{a}_2\,\cdots\,\mathbf{a}_n]$ be a matrix over a field $\mathbb{F}$. Let $E=\{1,2,\dots,n\}$, and define
  \begin{equation}
    \mathcal{I}\triangleq \bigl\{ X\subseteq E : \{\mathbf{a}_i:i\in X\}\ \text{is linearly independent over }\mathbb{F} \bigr\}.
  \end{equation}
  Then $(E,\mathcal{I})$ is a matroid.
\end{proposition}

\begin{theorem}
Let $A=[\mathbf{a}_1\,\mathbf{a}_2\,\cdots\,\mathbf{a}_n]\in\mathbb{F}^{m\times n}$ have rank $r$, and let
$c_1,\ldots,c_n\in\mathbb{R}$ be the column costs (e.g., the LLR weights). Consider the greedy procedure that orders the
indices so that
\begin{equation}
    c_{\pi(1)}\le c_{\pi(2)}\le \cdots \le c_{\pi(n)}
\end{equation}

and then scans $\pi(1),\pi(2),\ldots,\pi(n)$ in this order, adding index $\pi(t)$ whenever
\begin{equation}
    \{\mathbf{a}_j:j\in G\cup\{\pi(t)\}\}
\end{equation}
is linearly independent, where $G$ denotes the set of indices selected so far. Let $G^\star$ be the
set obtained after $r$ columns have been selected. Then $G^\star$ is a minimum-cost basis of the
column matroid of $A$, i.e.,
\begin{align}
\sum_{j\in G^\star} c_j &= \min\Bigl\{\sum_{j\in B} c_j:\, B\subseteq E, \nonumber \\
&\qquad\qquad\quad \{\mathbf{a}_j:j\in B\}\text{ is a basis of }\operatorname{col}(A) \Bigr\}.
\end{align}
\end{theorem}

\begin{proof}
Let $E=\{1,2,\dots,n\}$ and define
\begin{equation}
    \mathcal{I}
=
\Bigl\{
I\subseteq E:\{\mathbf{a}_j:j\in I\}\text{ is linearly independent}
\Bigr\}.
\end{equation}
As established in Proposition~\ref{prop:vector-matroid}, $(E,\mathcal{I})$ is the vector matroid associated with $A$.
Its rank is
\begin{equation}
    \rho(E)=\operatorname{rank}(A)=r.
\end{equation}
Hence the bases of this matroid are exactly the index sets $B\subseteq E$ of cardinality $r$ such
that $\{\mathbf{a}_j:j\in B\}$ is a basis of $\operatorname{col}(A)$.

Now define a shifted weight function $w:E\to\mathbb{R}$ by
\begin{equation}
    w_j = C-c_j,\qquad j\in E,
\end{equation}
where
\begin{equation}
    C=\max_{1\le j\le n} c_j.
\end{equation}
Then $w_j\ge 0$ for all $j\in E$. Since every basis $B$ of the matroid has cardinality $r$, we have
\begin{equation}
    \sum_{j\in B} w_j
=
\sum_{j\in B}(C-c_j)
=
rC-\sum_{j\in B} c_j.
\end{equation}
Therefore, for two bases $B_1$ and $B_2$,
\begin{equation}
    \sum_{j\in B_1} c_j \le \sum_{j\in B_2} c_j
\quad\Longleftrightarrow\quad
\sum_{j\in B_1} w_j \ge \sum_{j\in B_2} w_j.
\end{equation}
Thus minimizing the total cost over all bases is equivalent to maximizing the total shifted weight
over all bases.

The greedy algorithm described in the statement scans the columns in nondecreasing order of $c_j$.
Equivalently, it scans them in nonincreasing order of $w_j$, since
\begin{equation}
    c_i\le c_j \quad\Longleftrightarrow\quad w_i\ge w_j.
\end{equation}
At each step it keeps an index precisely when adding that index preserves independence in the
matroid $(E,\mathcal{I})$. Hence this is exactly the standard greedy algorithm for finding a
maximum-weight basis of a matroid.

By the classical theorem of Edmonds on matroids and the greedy algorithm
\cite{edmonds1971matroids}, the greedy algorithm returns a basis $G^\star$ maximizing
$\sum_{j\in B} w_j$ over all bases $B$ of $(E,\mathcal{I})$. By the equivalence proved above, the
same basis $G^\star$ minimizes $\sum_{j\in B} c_j$ over all bases. This proves the claim.

\end{proof}

\bibliographystyle{IEEEtran}
\bibliography{bibliography}

\end{document}